\documentclass[letterpaper, 10 pt, conference]{ieeeconf}  
\IEEEoverridecommandlockouts  
\UseRawInputEncoding

\usepackage{color}
\usepackage{amsmath}
\usepackage{amsfonts}

\usepackage{bm}
\usepackage{mathtools}
\usepackage{cite}

\usepackage{amsthm}

\usepackage{subcaption}
\usepackage{hyperref}
\usepackage{todonotes}

\newtheorem{problem}{Problem} 
\newtheorem{definition}{Definition} 
\newtheorem{theorem}{Theorem} 
\newtheorem{proposition}{Proposition} 
\newtheorem{lemma}{Lemma} 

\usepackage[linesnumbered,ruled,vlined]{algorithm2e}

\newcommand{\domain}[0]{X_{\text{safe}}}
\newcommand{\region}[1]{\mathfrak{r}_{#1}}
\newcommand{\cov}[1]{\text{Cov}_{\bm{#1}}}
\newcommand{\nn}[1]{
\ifthenelse{\equal{#1}{}}{f^{w}_a}{f^{w}_a(#1)}}

\newcommand{\nnT}[1]{\mathcal{T}f_a^{w}(\mathcal{T}^{-1}#1)}
\newcommand{\tk}[1]{T^a(#1|x)}
\newcommand{\x}[1]{\bm{x}_{#1}}
\newcommand{\noise}[1]{\bm{v}_{#1}}
\newcommand{\normaldist}[3]{\mathcal{N}(#1\mid #2, #3)}
\newcommand{\pathx}[1]{\omega^{#1}_{\bm{x}}}
\newcommand{\pathsx}[0]{\Omega_{\bm{x}}}
\newcommand{\finitepathsx}[0]{\pathsx^{\text{fin}}}
\newcommand{\strategyx}[0]{\pi_{\bm{x}}}
\newcommand{\strategiesx}[0]{\Pi_{\bm{x}}}
\newcommand{\prop}[1]{\mathfrak{p}_{#1}}

\newcommand{\argmax}[1]{\underset{#1}{\text{arg max}\ }}
\newcommand{\argmin}[1]{\underset{#1}{\text{arg min}\ }}
\newcommand{\erf}[1]{\text{erf}\left(#1\right)}
\newcommand{\g}[0]{g(z)}
\newcommand{\T}[0]{\mathcal{T}}

\newcommand{\lbf}[0]{\check{f}}
\newcommand{\ubf}[0]{\hat{f}}

\newcommand{\post}[1]{Post(#1,\mathcal{T},a)}
\newcommand{\image}[1]{Im(#1, \mathcal{T})}
\newcommand{\conv}[1]{conv(#1)}
\newcommand{\rect}[1]{rect(\check{#1},\hat{#1})}
\newcommand{\vertices}[1]{vert(#1)}

\newcommand{\postapprox}[1]{\widetilde{Post}(#1,\mathcal{T},a)}
\newcommand{\rectf}[1]{rect(\lbf(#1),\ubf(#1))}
\newcommand{\rectregion}[1]{rect(#1)}
\newcommand{\rectv}[2]{rect(#1,#2)}

\newcommand{\imdp}[0]{\mathcal{I}}
\newcommand{\Q}[0]{Q_s}

\newcommand{\pathI}[0]{\omega_{\mathcal{I}}}
\newcommand{\finitepathI}[0]{\omega^{\text{fin}}_{\mathcal{I}}}
\newcommand{\pathsI}[0]{\Omega_{\mathcal{I}}}
\newcommand{\finitepathsI}[0]{\pathsI^{\text{fin}}}
\newcommand{\strategyI}[0]{\pi_{\mathcal{I}}}
\newcommand{\strategiesI}[0]{\Pi_{\mathcal{I}}}
\newcommand{\adversary}[0]{\xi}
\newcommand{\adversaries}[0]{\Xi}
\newcommand{\propositionsI}[0]{AP_{\imdp}}
\newcommand{\LabelingI}[0]{L_{\imdp}}

\newcommand{\imdpdef}[0]{\imdp=(Q,A,\check{P}, \hat{P}, \propositionsI, \LabelingI)}

\newcommand{\uP}[2]{\hat{P}(#1,a,#2)}
\newcommand{\lP}[2]{\check{P}(#1,a,#2)}

\newcommand{\satlp}[1]{\check{p}(#1)}
\newcommand{\satup}[1]{\hat{p}(#1)}

\newcommand{\mappingpoint}[0]{M_{\bm{x}}}
\newcommand{\mappingpath}[0]{M_{\omega}}

\usepackage{ifthen}
\newboolean{arxiv}

\usepackage{array}
\usepackage{booktabs}

\title{\LARGE \bf
Formal Control Synthesis for Stochastic Neural Network Dynamic Models
}

\author{Steven Adams$^{1}$,  Morteza Lahijanian$^{2}$, and Luca Laurenti$^{1}$
\thanks{$^{1}$Delft Center for Systems and Control,  TU Delft}
\thanks{$^{2}$Dept. of Aerospace Eng. Sciences and Computer Science, CU Boulder}
}

\begin{document}
\bstctlcite{IEEEexample:BSTcontrol}
\setboolean{arxiv}{true}

\maketitle
\thispagestyle{empty}
\pagestyle{empty}


\maketitle

\begin{abstract}
Neural networks (NNs) are emerging as powerful tools to represent the dynamics of control systems with complicated physics or black-box components.  Due to complexity of NNs, however, existing methods are unable to synthesize complex behaviors with guarantees for \textit{NN dynamic models} (NNDMs).
This work introduces a control synthesis framework for stochastic NNDMs with performance guarantees.  The focus is on specifications expressed in \textit{linear temporal logic interpreted over finite traces} (LTLf), and the approach is based on finite abstraction.  Specifically, we leverage recent techniques for convex relaxation of NNs to formally abstract a NNDM into an \textit{interval Markov decision process} (IMDP). Then, a strategy that maximizes the probability of satisfying a given 
specification is synthesized over the IMDP and mapped back to the underlying NNDM.  
We show that the process of abstracting NNDMs to IMDPs reduces to a set of convex optimization problems, hence guaranteeing efficiency. 
We also present an adaptive refinement procedure that makes the framework scalable. 
On several case studies, we illustrate the our framework is able to provide non-trivial guarantees of correctness for NNDMs with architectures of up to 5 hidden layers and hundreds of neurons per layer.

\end{abstract}



\section{Introduction}\label{sec:introduction}


\noindent
Autonomous systems are becoming increasingly complex, often including black-box components and performing complex tasks in the presence of uncertainty. In this context, because of their data efficiency and representation power, deep neural networks (NNs) 
can be a transformative technology: NNs have already achieved state-of-the-art performance to model and control dynamical systems in various fields, including reinforcement learning (RL) \cite{nagabandi2018neural}. 
However, employing NNs in \emph{safety-critical} applications, such as UAVs, 
where failures may have catastrophic effects,
remains a major challenge due to limitations of existing methods to provide performance guarantees. 
This work focuses on this challenge and 
develop a correct-by-construction synthesis framework for systems 
with NN dynamic models.


To achieve complex behaviors with strong guarantees, formal synthesis for control systems have been well-studied in recent years \cite{Tabuada2009VerificationApproach,Lahijanian2015FormalSystems,Doyen2018VerificationSystems}. 
These methods use expressive formal languages such as \textit{linear temporal logic} with infinite (LTL) \cite{baier2008principles} or \textit{finite} (LTLf) \cite{LTLf} interpretation over traces, to specify complex behaviors, and apply rigorous techniques to abstract the dynamics to  finite (Kripke) models.  Then, by utilizing model-checking-like algorithms on the abstraction, they synthesize controllers that achieve the specification.  
The key step in these methods is the abstraction construction, which often relies on (simple) analytical models. For modern systems, however, such models are often unavailable due to, e.g., complexity of the physics or black-box components.


To describe complex  dynamical dynamical systems, 
 NNs have been already used with success \cite{raissi2019physics,gamboa2017deep}.
Furthermore,
 the ability of \textit{NN dynamic models} (NNDMs) to predict complex dynamics has also been employed to enhance controller training in RL frameworks \cite{nagabandi2018neural,chua2018deep}. In these works, a NN model of the system is trained in closed-loop with a NN controllers, which can be concatenated in a  single NN representing the dynamics of the closed-loop system.  
These benefits have motivated the recent development of methods for formal analysis of NNDM  properties~\cite{Wei2021SafeModels,wicker2021certification,fazlyab2021introduction},
extending  verification algorithms for NNs \cite{Xu2020AutomaticBeyond} to support temporal properties.  
Nevertheless, these methods are limited to simple safety properties and often neglect noise in the dynamics. As a consequence, the state-of-the-art techniques for NNDM are still unable to achieve complex behaviors with guarantees. 


In this work, we close the gap by introducing a control synthesis framework for stochastic NNDMs to achieve a complex specification with formal guarantees. Our approach is based on finite abstraction, and we use 
LTLf
as the specification language which has the same expressively as LTL, but specifies finite behaviors, making it an appropriate language for stochastic models. 
In particular, we leverage recent convex relaxation techniques for NNs \cite{Xu2020AutomaticBeyond} to build piece-wise linear functions that under- and overapproximate the NNDM and construct the abstraction as an interval Markov decision process (IMDP) \cite{IMDP}.
Critically, we show that this discretization-based method only requires solving a set of convex optimization problems, which can be reduced to evaluation of an analytical function on a finite set of points, resulting in efficient abstraction procedure. 
Then, we use existing tools to synthesize a control strategy that optimizes the probability of satisfying a given specification while guaranteeing robustness against uncertainties due to dynamics approximation and discretization. To ensure scalability, we present an adaptive refinement algorithm that iteratively reduces uncertainty in a targeted manner. 
Finally, we illustrate the effecacy of our framework in several case studies.

In summary, the contributions of this paper are:
(i) a novel framework for formal synthesis for stochastic NNDMs with complex specifications,
(ii) an efficient finite abstraction technique for NNDMs,
(iii) an adaptive refinement algorithm for uncertainty reduction, and
(iv) illustration of the efficacy and scalability of the framework on a set of rich case studies with complex NNDMs, whose architecture include up to five hidden layers and hundreds of neuron per layer.



\section{Problem Formulation}


We consider the following stochastic \emph{neural network dynamic model (NNDM)}: 
\begin{equation} 
    \label{eq: process}
    \x{k+1} = \nn{\x{k}} + \noise{k},  \quad \noise{k} \sim \mathcal{N}(0,\cov{v}),
\end{equation}
where $k\in\mathbb{N}$, $\x{k}, \noise{k} \in \mathbb{R}^n$ and $a \in A = \{a_1,\hdots,a_m\}$ is a finite set of actions. For every $a$, $\nn{}:\mathbb{R}^n\rightarrow \mathbb{R}^n$ is a (trained) 
feed-forward NN with ReLU, sigmoid or tanh activation functions, where $w$ denotes the maximum likelihood weights.
The noise term $\noise{k}$ is a random variable with stationary 
Gaussian distribution 
with zero mean and covariance $\cov{v}\in\mathbb{R}^{n\times n}$. 
Intuitively, $\x{k}$ is a discrete-time stochastic process whose time evolution is given by iterative predictions of various NNs.  
We remark that models such as Process \eqref{eq: process} are increasingly employed in both robotics and biological systems for both model representation and NN controller training with, e.g., state-of-the-art model-based RL techniques \cite{nagabandi2018neural,Zhao2020LearningCertificates,raissi2019physics}. For instance, Process \eqref{eq: process} can represent a NN model in closed loop with different (possibly NN) feedback controllers, and the role of actions $a$ is to switch between different controllers.

Let $\pathx{N}=x_0\xrightarrow{a_0}x_1\xrightarrow{a_1}\hdots \xrightarrow{a_{N-1}}x_N$ be a finite path of Process \eqref{eq: process} of length $N\in \mathbb{N}$ and $\finitepathsx$ be the set of all finite paths.
Paths of infinite length and the set of all paths of infinite lengths are denoted by $\pathx{}$ and $\pathsx$, respectively, with $\pathx{}(k)$ denoting the state of $\pathx{}$ at time $k$.
Given a finite path, a \textit{switching strategy} $\strategyx:\finitepathsx\rightarrow A$ chooses the next action of Process \eqref{eq: process}. The set of all switching strategies is denoted by $\strategiesx$.  
For $a\in A$, $X\subseteq \mathbb{R}^n$, and $x\in \mathbb{R}^n$, we call 
\begin{equation}\label{eq: transition kernel}
    \tk{X}=\int_X \normaldist{\bar{x}}{\nn{x}}{\cov{v}}d\bar{x}
\end{equation}
the \textit{transition kernel} of Process \eqref{eq: process} under action $a$, where $\normaldist{\cdot}{\nn{x}}{\cov{v}}$ is a normal distribution with mean $\nn{x}$ and covariance $\cov{v}$.
For a strategy $\strategyx$, Process \eqref{eq: process} defines a probability measure $P$ which is uniquely defined by $T^a$ and by the initial conditions \cite{Bertsekas1996StochasticCase} s.t. 
for every $k>0$,
\begin{equation*}
    P[\pathx{}(k+1)\in X \mid \pathx{}(k)=x,\strategyx(\pathx{}(k)) = a]=\tk{X}.
\end{equation*}

We are interested in the behavior of Process \eqref{eq: process} in compact set $\domain\subset\mathbb{R}^n$ with respect to the regions of interest in $R=\{\region{1}, \hdots,\region{|R|}\}$, where $\region{i} \subseteq \domain$. To define properties over $R$, we associate to each region $\region{i}$ the atomic proposition $\prop{i}$ such that $\prop{i}$ is true iff 
$x\in \region{i}$. The set of atomic propositions is given by $AP$, and the labeling function $L:X\rightarrow 2^{AP}$ returns the set of atomic propositions that are true at each state. Then, we define the observation of path $\pathx{N}$ to be $\rho=\rho_0\rho_1\hdots\rho_N$, where $\rho_i=L(\pathx{N}(i))$ for all $i\leq N$. 

To express the temporal properties of Process \eqref{eq: process}, we consider 
LTLf, which is an expressive language to specify finite behaviors, and hence, appropriate for stochastic systems.

\begin{definition}
An LTLf formula is built from a set of propositional symbols $AP$ and is closed under the boolean connectives as well as 
the ``until" operator $\mathcal{U}$, and the temporal ``eventually'' $\mathcal{F}$ and ``globally'' $\mathcal{G}$ operators:
\begin{equation*}
    \phi \coloneqq \top \mid \prop{} \mid \lnot \phi \mid \phi_1 \wedge \phi_2 \mid \phi_1\mathcal{U}\phi_2 \mid \mathcal{F}\phi \mid \mathcal{G} \phi
\end{equation*}
where 
$\prop{}\in AP$. 
\end{definition}
The semantics of LTLf can be found in \cite{LTLf}. We say a path $\pathx{}$ satisfies $\phi$, denoted by $\pathx{} \models \phi$, if a prefix of its observation satisfies $\phi$ \cite{ltlf_morteza}.

\begin{problem}[Control synthesis]\label{prob: synthesis}
Given a NNDM as defined in Process \eqref{eq: process}, a compact set $\domain$, and an LTLf formula $\phi$ defined over the regions of interest in $\domain$, find a switching strategy $\strategyx^*$ that maximizes the probability that a path $\pathx{} \in \pathsx$ of Process \eqref{eq: process} satisfies  $\phi$ while remaining in $\domain.$
\end{problem}
To solve Problem \ref{prob: synthesis}, we abstract Process \eqref{eq: process} into a finite Markov model, where the stochastic nature of Process \eqref{eq: process} and the error corresponding to 
the discretization of the space 
are formally modelled as uncertainties. The abstracting process involves the computation of bounds on the transition probabilities between different regions of the state space. 
In section \ref{sec: computation bounds by crown}, we show that by using linear functions that locally under and overapproximate the NN-dynamics, these bounds can be efficiently computed by solving convex optimization problems. 
For the resulting Markov model, we then synthesize a strategy that maximizes the probability that the paths of the Markov model satisfy $\phi$ and can be mapped onto Process \eqref{eq: process}. 
Finally, in Section \ref{sec: synthesis driven refinement} 
we develop a refinement scheme that iteratively builds a finer abstraction based on the synthesis results by reducing the conservatism induced by the approximation bounds of $\nn{}$.

\section{Preliminaries}
\paragraph*{Notation}
We denote by $x^{(l)}$ the $l$-th element of vector $x\in\mathbb{R}^n$.
Further, for convex region $X\subset \mathbb{R}^n$, we denote by $X^{[l]}\subset \mathbb{R}$ the interval of values of $X$ in the $l$-th dimension, i.e., $
    X^{[l]} = \{x^{(l)} \mid x\in X\}$.
Given a linear transformation function (matrix) $\T\in \mathbb{R}^{n \times n}$, the image of region $X\subset\mathbb{R}^n$ under $\T$ is defined as $
    \image{X} = \{\T x \mid x\in X\}.$
The post image of region $X$ under action $a$ of Process~\eqref{eq: process} and $\T$ is defined as $
    Post(X,\T,a) = \{\T\nn{x} \mid x\in X\}.$ 

For vectors $x,x'\in\mathbb{R}^n$, we denote by $\rectv{x}{x'}$ the axis-aligned hyper-rectangle that is defined by the intervals $r_1 \times r_2 \times \hdots \times r_n$, where $
 r_l = [\min(x^{(l)}, x'^{(l)}), \max(x^{(l)}, x'^{(l)})]$. 
In addition, for region $X\subset \mathbb{R}^n$, we denote by $\rectregion{X}$ the hyper-rectangular overapproximation of $X$, i.e., $
    \rectregion{X} = \rect{x}, $
where $\check{x}^{(l)}=\inf(X^{[l]})$ and $\hat{x}^{(l)}=\sup(X^{[l]})$ for every $l\in \{0,\hdots,n\}$. 
Lastly, we define a \textit{proper linear transformation function} as follows.  
\begin{definition}
For $X \subset \mathbb{R}^n$, the transformation matrix $\T \in \mathbb{R}^{n\times n}$ is proper w.r.t. $X$ if $Im(X,\T)$ is an axis-aligned hyper-rectangle.
\end{definition}
\noindent
Note that, as $\T$ is a linear transformation, $X$ must necessarily be a convex polytope in order for $\T$ to be proper \cite{rockafellar2015convex}.

\paragraph*{Interval Markov Decision Processes}
We utilize Interval Markov Decision Processes (IMDP), also called Bounded MDPs \cite{IMDP}, 
to abstract Process \eqref{eq: process}. IMDPs are a generalized class of MDPs
that allows for a range of transition probabilities between states.

\begin{definition}
An \textit{interval Markov decision process} (IMDP) is a tuple $\imdpdef$, where
\begin{itemize}
    \item $Q$ is a finite set of states, 
    \item $A$ is a finite set of actions available in each state $q\in Q$.
    \item $\check{P}:Q\times A \times Q\rightarrow  [0,1]$ is a function, where $\lP{q}{q'}$ defines the lower bound of the transition probability from state $q\in Q$ to state $q'\in Q$ under action $a\in A$,
    \item $\hat{P}:Q\times A \times Q \rightarrow [0,1]$ is a function, where $\uP{q}{q'}$ defines the upper bound of the transition probability from state $q\in Q$ to state $q'\in Q$ under action $a\in A$.
    \item $\propositionsI$ is a finite set of atomic propositions,
    \item $\LabelingI: Q\rightarrow 2^{AP}$ is a labeling function assigning to each state $q\in Q$ a subset of $\propositionsI$.
\end{itemize}
\end{definition}

For all $q,q' \in Q$ and $a \in A$, it holds that $\lP{q}{q'}\leq\uP{q}{q'}$ and $\sum_{q'\in Q}\lP{q}{q'}\leq 1 \leq \sum_{q' \in Q}\uP{q}{q'}$. 
A path $\pathI$ of an IMDP is a sequence of states $\pathI=q_0 \xrightarrow{a_0} q_1 \xrightarrow{a_1} q_2 \xrightarrow{a_2} \hdots$ such that $\uP{q_k}{q_{k+1}}>0$ for all $k\in \mathbb{N}$. We denote the last state of a finite path $\finitepathI$ by $last(\finitepathI)$ and the set of all finite and infinite paths by $\finitepathsI$ and $\pathsI$, respectively.  
A strategy of an IMDP $\strategyI: \finitepathsI \rightarrow A$ maps a finite path $\finitepathI\in\finitepathsI$ of $\imdp$ onto an action in $A$. The set of all strategies is denoted by $\strategiesI$. 
Let $\mathcal{D}(Q)$ denote the set of discrete probability distributions over $Q$. Given a strategy $\strategyI$, the IMDP reduces to a set of infinitely many Markov chains defined by the transition probability bounds of the IMDP. An adversary chooses a feasible distribution from this set at each state and reduces the IMDP to a Markov chain. 

\begin{definition}
For an IMDP $\imdp$, an adversary is a function $\adversary: \finitepathsI \times A \rightarrow \mathcal{D}(Q)$ 
that, for each finite path $\finitepathI \in \finitepathsI$, state $q = last(\finitepathI)$, and action $a\in A$, assigns a feasible distribution $\gamma_q^a$ which satisfies $\lP{q}{q'} \leq \gamma_q^a(q') \leq \uP{q}{q'}$. The set of all adversaries is denoted by $\adversaries$.
\end{definition}


\section{IMDP Abstraction}\label{sec: imdp abstraction}
In order to solve Problem \ref{prob: synthesis}, we first abstract Process \eqref{eq: process} into an IMDP $\imdpdef$. 
To do that, similarly as in \cite{cauchi2019efficiency}, we discretize $\domain$ in such way that the transition kernel in \eqref{eq: transition kernel} can be computed analytically. 
Let $\T$ be the Mahalanobis transformation $\T= \Lambda^{-\frac{1}{2}}\mathcal{V}^T$, where $\Lambda = \mathcal{V}^T\cov{v}\mathcal{V}$ is a diagonal matrix whose entries are the eigenvalues of $\cov{v}$, and $\mathcal{V}$ is the corresponding orthogonal (eigenvector) matrix. Then, the distribution of $\T \x{k+1}$ given $\x{k}=x$ under action $a$ becomes $\normaldist{\cdot}{\T \nn{x}}{I}$, where $I$ is the identity matrix. Consequently, given a region $X \subset \mathbb{R}^n$ 
for which $\T$ is a proper transformation (i.e., $\image{X}=\rect{x}$),
we obtain that $\tk{X}=g(\T \nn{x})$, where
\begin{equation}\label{eq: g}
\begin{aligned}
    \g =
    \frac{1}{2^n} \prod_{l=1}^n \left(\text{erf}\left(\frac{z^{(l)}-\check{x}^{(l)}}{\sqrt{2}}  \right)-\text{erf}\left(\frac{z^{(l)}-\hat{x}^{(l)}}{\sqrt{2}}\right)\right),
\end{aligned}
\end{equation}
and $\erf{\cdot}$ is the error function.
Hence, we discretize $\domain$ by using a grid in $\image{\domain}$, and denote by $Q_s=\{q_1,\hdots, q_{|Q_s|}\}$ the resulting set of regions. To each cell $q_i$, we associate a state of the IMDP $\imdp$. We overload the notation by using $q_i$ for both a region in $\domain$ and a state of $\imdp$. Then, the set of states of $\imdp$ is defined as $Q=Q_s \cup \{q_u\}$, where $q_u$ denotes the remainder of the state space, i.e. $\mathbb{R}^n \setminus \domain$.

We define the set of actions of $\imdp$ to be the set of actions $A$ of Process \eqref{eq: process}. To ensure a correct abstraction of Process \eqref{eq: process}, we assume a discretization of $\domain$ that respects the regions of interest in $R$, i.e., $\forall r \in R, \exists Q_r \subseteq Q$ such that $\cup_{q\in Q_r} q=r$. Under this assumption, the set of atomic propositions $\propositionsI$ is equal to $AP$. We define the labeling function $\LabelingI$ with $\LabelingI(q)=L(x)$ for any choice of $x\in q$. 

To compute the transition probability bounds $\check{P}$ and $\hat{P}$ for all $q,q'\in \Q$ and $a\in A$, we need to derive the following bounds, which are the subject of Section \ref{sec: computation bounds by crown} and \ref{sec: efficient computation}:
\begin{align}\label{eq: optim problms trans kern}
    \lP{q}{q'} \leq \min_{x\in q} \tk{q'},\   
    \uP{q}{q'} \geq \max_{x\in q} \tk{q'} 
\end{align}
The probability interval for transitioning to the state $q_u \in Q$, i.e., the region outside of $\domain$, is given by
\begin{align*}
    \lP{q}{q_u} &\leq 1- \max_{x\in q} \tk{q_u}\\
    \uP{q}{q_u} &\geq 1- \min_{x\in q} \tk{q_u},
\end{align*}
for all $a \in A$ and $q \in \Q$. Since we are not interested in the behavior of Process \eqref{eq: process} outside $\domain$, we make $q_u$ absorbing, i.e. $\lP{q_u}{q_u}=\uP{q_u}{q_u}=1$ for all $a\in A$.

\subsection{Transition Probability Bounds Computation}\label{sec: computation bounds by crown}
\begin{figure}
    \centering
    \vspace{2mm}
    \includegraphics[width=0.35\textwidth]{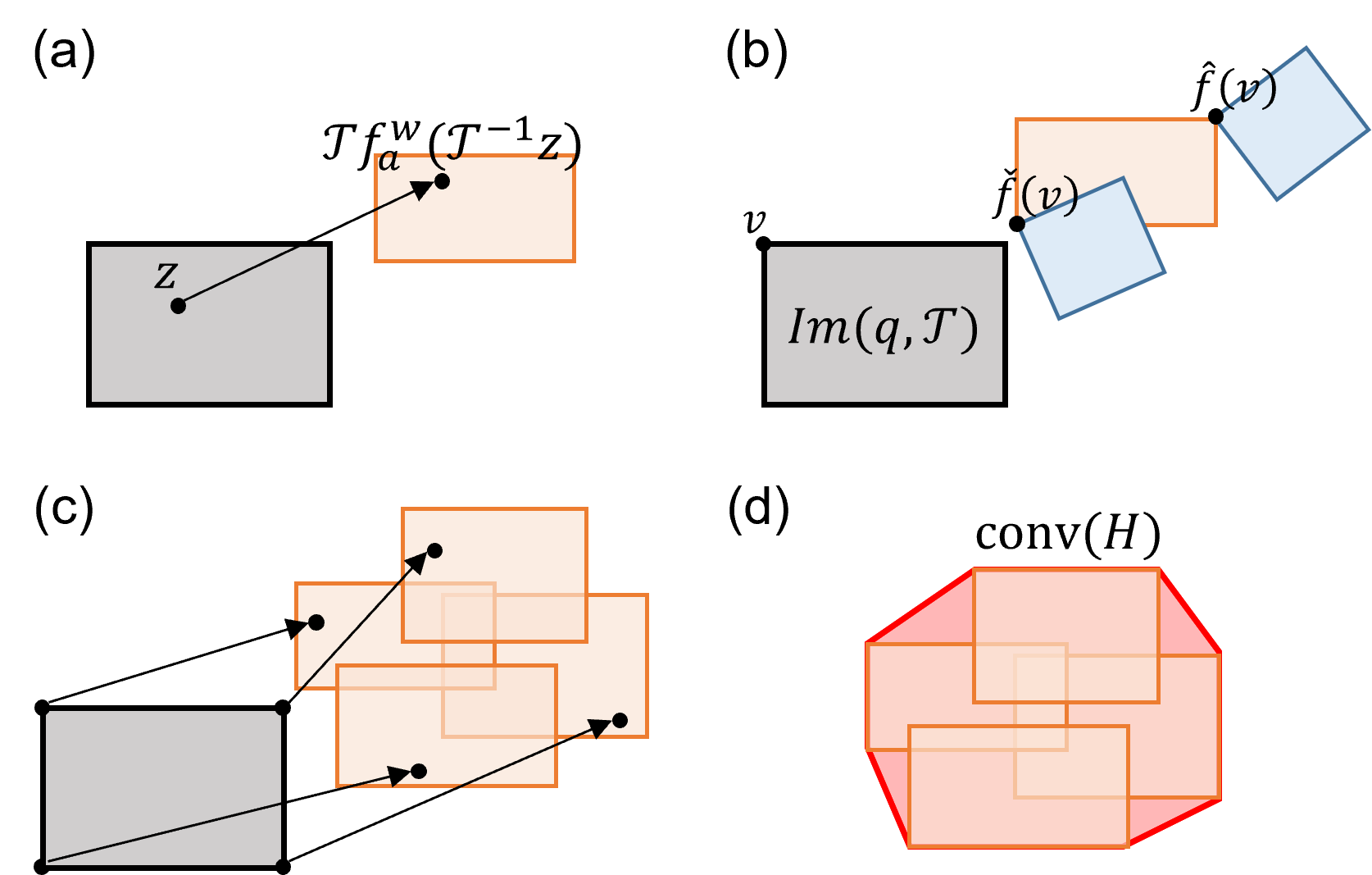}
    \caption{(a) For each point $z$ in the black rectangle ($\image{q}$), $\nnT{z}$ is contained in the orange rectangle whose vertices are defined by $\lbf(z)$ and $\ubf(z)$. (b) The blue rectangles contain $\ubf(z)$ and $\ubf(z)$ for all $z\in \image{z}$.
    (c) The orange rectangles capture $\nnT{v}$ for all vertices $v$ of $\image{q}$ and (d) fully define the red convex region that captures $\nnT{z}$ for all $z\in \image{q}$.
    }
    \vspace{-2mm}
    \label{fig:one-step crown set}
\end{figure}
We derive an efficient and scalable procedure for the computation of the bounds in \eqref{eq: optim problms trans kern}. 
Recall that the discretization procedure described above enables to write 
transition kernel $T^a$ as the product of {erf} in \eqref{eq: g}. Hence,
the optimization problems in \eqref{eq: optim problms trans kern} can be performed on 
\eqref{eq: g}, i.e., for $q',q \in \Q$ to bound $\tk{q}$, we can optimize $g(z)$ over $z\in\post{q}$. However, the 
exact computation of $\post{q}$ is intractable since
NN-dynamics are inherently nonconvex.  
Hence, we instead seek to overapproximate $\post{q}$ 
by recursively finding linear functions on the NN-structure that under- and overapproximate $\nn{x}$ for all $x\in q \in Q$, as shown in \cite{Xu2020AutomaticBeyond}. 
We can then utilize these linear functions to bound $\post{q}$ for all $q\in \Q$ as shown in the following proposition. 

\begin{proposition}\label{lemma: one-step crown set}
For Process \eqref{eq: process} under action $a$, region $q\subset \mathbb{R}^n$, and proper transformation matrix $\T$   w.r.t. $q$, let $\{v_1,\hdots,v_{(2^n)}\}\in\mathbb{R}^n$ be the vertices of hyper-rectangle $\image{q}$, and $\lbf, \ubf: \mathbb{R}^n \rightarrow \mathbb{R}^n$ be linear functions that bound $\nnT{z}$ for all $z\in \image{q}$. Define
$  H = \{\rectf{v}\mid v\in \{v_1,\hdots,v_{(2^n)}\}\}.$
Then, it holds that 
$ \post{q} \subseteq \conv{H}, $
where $\conv{H}$ is the convex hull of hyper-rectangles in $H$.
\begin{proof}
For $z\in \image{q}$, we have that $    \nnT{z} \in\rectf{z}$.
 Consequently,   $\post{q} \subseteq \postapprox{q}=\cup_{z\in \image{q}} \rectf{z}   $. Note that $\image{q}$ is a convex-polytope and the construction of $\rectf{z}$ only involves linear operations. As a consequence, $\postapprox{q}$ is fully described by the vertices of $\image{q}$. Hence, $\postapprox{q}{\subseteq} 
 conv(H)$.
\end{proof}
\end{proposition}
As a result of the above proposition, to construct the post image overapproximation induced by the local linear under- and overapproximations of the NN, we only have to check the vertices of the image as illustrated in Figure \ref{fig:one-step crown set}. 
Utilizing the analytical reformulation of the transition kernel as in \eqref{eq: g} and the post-image overapproximation as derived in Proposition \ref{lemma: one-step crown set}, we obtain that 
\begin{align}
    \min_{x\in q}\tk{q'} &\geq\min_{z\in \conv{H}} \g  \label{eq: exact min opt prob} \\
    \max_{x\in q}\tk{q'} &\leq \max_{z\in \conv{H}} \g  \label{eq: exact max opt prob}
\end{align}
where $H$ is as defined in Proposition \ref{lemma: one-step crown set}. Here, \eqref{eq: exact max opt prob} is a log-concave maximization problem, which can be solved with standard convex optimization algorithms, such as gradient descent \cite{Boyd2004ConvexOptimization}. Although \eqref{eq: exact min opt prob} is in general non-convex, the following result -- a consequence of Corollary 32.3.4 in \cite{rockafellar2015convex} -- guarantees that to compute lower bound \eqref{eq: exact min opt prob}, i.e., the minimum of a log-concave problem, it suffices to check the vertices of $\conv{H}$. 

\begin{lemma}
\label{prop: min log-concave func}
For $\g$ as defined in \eqref{eq: g} it holds that
\begin{equation}
    \min_{z\in\conv{H}} \g = \min_{z\in V} \g  
\end{equation}
where $V$ is the set of vertices of $\conv{H}$.
\end{lemma}


\subsection{Efficient Computation of Transition Probabilities}\label{sec: efficient computation}
Note that, although solving for \eqref{eq: exact min opt prob} and \eqref{eq: exact max opt prob} reduces to the solution of convex maximization and minimization problems, to build the abstraction, we still need to solve  $\mathcal{O}(|\Q|^2)$ of these problems (one for each pair of states in $\Q$). 
This becomes expensive for large $|\Q|$, which is often the case for high-dimensional systems.
In this section, we propose an alternative approach to reduce this computational burden. 
In particular, the following theorem shows that if we overapproximate $conv(V)$ by an axis-aligned hyper-rectangle, to find solutions to \eqref{eq: optim problms trans kern}, we
only have to check a finite number of points at the boundary of the axis-aligned hyper-rectangle and perform $\mathcal{O}(|\Q|^2)$ function evaluations. 

\begin{theorem}\label{theor: sorting}
{
For Process \eqref{eq: process} under action $a$, regions $q,q'\subset \mathbb{R}^n$, and proper transformation matrix $\T$ w.r.t. $q,q'$, construct $H$ w.r.t. $q$ per Proposition \ref{lemma: one-step crown set}. Further, let vectors $\check{z},\hat{z}$ define the vertices of $\rectregion{\conv{H}}$, i.e., $\rectregion{\conv{H}}= \rect{z}$ such that for every  $l \in \{0,\hdots,n\}$  $\check{z}^{(l)}\leq\hat{z}^{(l)}$, and  denote by $\bar{v}$ 
the center of $\image{q'}$. 
Then, for $z_{min},z_{max}\in \mathbb{R}^n$ defined such that for $l \in \{0,\hdots,n\}$ 
\begin{align*}
    z_{min}^{(l)}&=\argmax{z^{(l)}\in \{ \check{z}^{(l)}, \hat{z}^{(l)}\}} |z^{(l)}-\bar{v}^{(l)}|,\\
    z_{max}^{(l)} &=  
    \begin{cases}
        \bar{v}^{(l)} &  \text{if }   \bar{v}^{(l)}\in  [\check{z}^{(l)}, \hat{z}^{(l)}]\\
        \argmin{z^{(l)}\in \{ \check{z}^{(l)}, \hat{z}^{(l)}\}} |z^{(l)}-\bar{v}^{(l)}| & \text{otherwise}
    \end{cases}
\end{align*}
it holds that 
\begin{align}
    \min_{x\in q}\tk{q'} \geq g(z_{min}),\quad
    \max_{x\in q}\tk{q'} \leq g(z_{max}). \label{eq: theor lb}
\end{align}
}
\end{theorem}
\begin{proof}
We consider the max case; the min case follows  similarly. By construction $\conv{H} \subseteq\rectregion{\conv{H}}$, hence 
$\max_{x\in q}\tk{q'} \leq \max_{z\in \rectregion{\conv{H}}} g(z). $ 
As $\rectregion{\conv{H}}$ is an axis-aligned hyperrectangle, it holds that
$  \max_{z\in \rectregion{\conv{H}}} g(z) = 
    \prod^n_{l=1} \max_{z^{(l)} \in [\check{z}^{(l)}, \hat{z}^{(l)}]} \int_{\check{v}^{(l)}}^{\hat{v}^{(l)}} \mathcal{N}(\bar{x}^{(l)}\mid z^{(l)}, 1)d\bar{x}^{(l)},$
where $[\check{v}^{(l)},\hat{v}^{(l)}]$ is the interval of the $l-$th dimension of $\image{q'}$.
This is a product of $n$ maximization problems that seek for the mean of a Gaussian distribution that maximizes its integral on a set. Each of these is maximized by minimizing the distance of $z^{(l)}$ to the center point $\bar{v}^{(l)}$ of the integration set. Hence, $z_{max}^{(l)}$ is equal to $\bar{v}^{(l)}$ if $\bar{v}^{(l)}\in [\check{z}^{(l)}, \hat{z}^{(l)}]$, else to one of the  endpoints of $[\check{z}^{(l)}, \hat{z}^{(l)}].$
\end{proof}

According to the above theorem, given $\conv{H}$ there exist a finite number
of potential $z_{max}$ and $z_{min}$. Moreover, given a potential $z_{min}$ or $z_{max}$ we can immediately find all the regions that take optimal value for \eqref{eq: theor lb} at this point, based on the positions of the regions in the grid, as illustrated in Figure \ref{fig: sorting}. As a consequence, we can simply check the finite sets of all possible $z_{min}$ and $z_{max}$ once to obtain $z_{max}$ and $z_{min}$ for all regions in the same group, and compute \eqref{eq: theor lb} by evaluating function $g$ on $z_{min}$ and $z_{max}$ for each region. Although this dramatic reduction of computation comes at the cost of more conservative bounds compared to solving \eqref{eq: exact min opt prob} and \eqref{eq: exact max opt prob}, the introduced overapproximation can be reduced by a refinement algorithm as proposed in Section \ref{sec: synthesis driven refinement}.

\begin{figure}
    \centering
    \vspace{2mm}
    \includegraphics[width=0.35\textwidth]{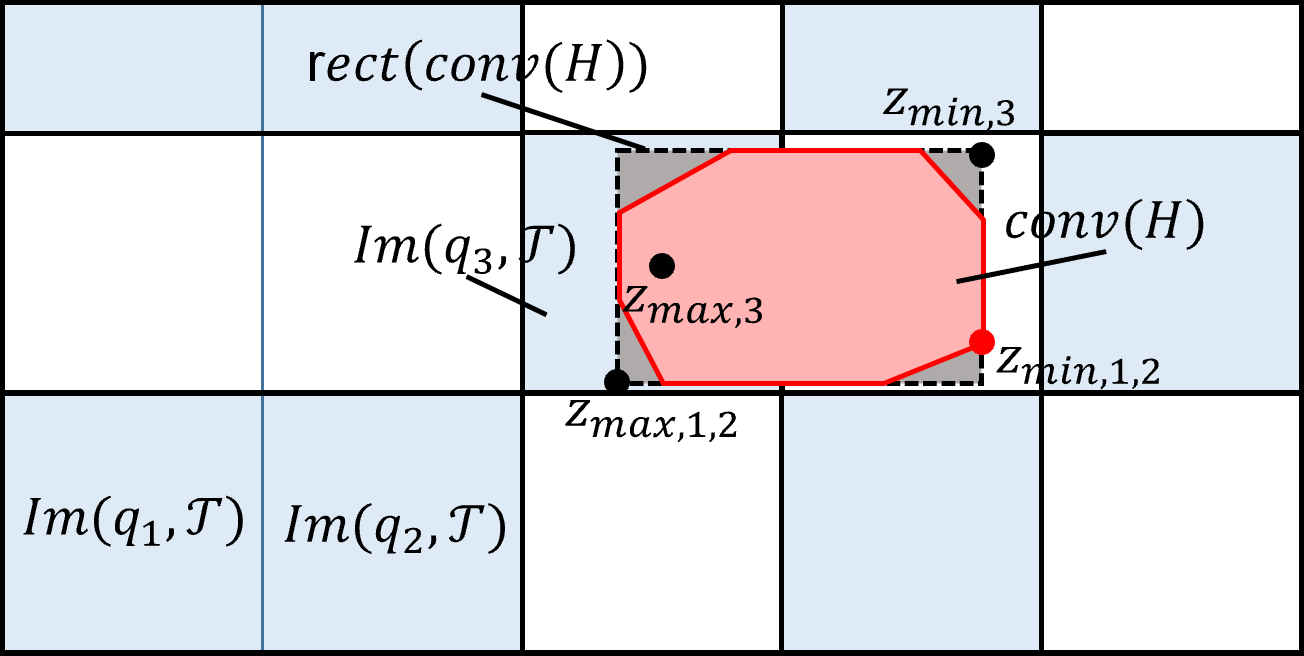}
    \caption{Regions $q_1,q_2$ share the same relative position w.r.t. $\rectregion{\conv{H}}$, i.e., $\forall z\in q_1,q_2$ it holds that $z^{(l)}\leq \inf(\rectregion{\conv{H}}^{[l]})$
    for all $l\in\{0,\hdots,n\}.$
    Consequently, they share 
    the same $z_{min}$ and $z_{max}$ for problems \eqref{eq: theor lb}. The white-blue coloring of the regions represents the grouping of regions accordingly.
    For overlapping region $q_3$, the minimizing location for \eqref{eq: exact min opt prob} is found at a vertex of $\conv{H}$.
    }
    \label{fig: sorting}
    \vspace{-1mm}
\end{figure}

\section{Control Synthesis \& Abstraction Refinement}\label{sec: synthesis}
Given Process \eqref{eq: process} and an LTLf property $\phi$, our objective is to synthesize a strategy that maximizes the probability of satisfying $\phi$. 
The IMDP abstraction $\imdp$ as constructed above, captures 
the behavior of Process \eqref{eq: process} w.r.t. the regions of interest. Therefore, we can focus on finding a strategy for $\imdp$ that maximizes $\phi$ subject to being 
robust against all the uncertainties (errors) induced by the discretization of space and the NN-dynamics approximation process.
This translates to assuming that the adversary's (uncertainty) objective is to minimize the probability of satisfaction.  Hence,
\begin{equation}\label{eq: synthesis problem}
    \strategyI^* = \argmax{\strategyI \in \strategiesI} \min_{\adversary \in \adversaries} \mathbb{P}[\pathI\models \phi \mid \strategyI, \adversary, \pathI(0)=q].
\end{equation}
We note that $\strategyI^*$ can be computed using known algorithms with a computational complexity polynomial in the number of states in the IMDP \cite{Lahijanian2015FormalSystems}.
To show that $\strategyI^*$ maps to a robust strategy for Process \eqref{eq: process}, we need to introduce a mapping between the process and the IMDP. Let 
$\mappingpoint: \mathbb{R}^n \rightarrow Q$
be a function that maps continuous states $x\in \mathbb{R}^n$ to their corresponding discrete regions in $Q$, i.e., $
x\in q \implies \mappingpoint(x)=q . $ 
In addition, let $\mappingpath:\finitepathsx\rightarrow\finitepathsI$ be a function that maps finite paths of Process \eqref{eq: process} to the finite paths of IMDP $\imdp$, i.e., for a finite path $\pathx{N} =x_0\xrightarrow{a_0}x_1\xrightarrow{a_1}\hdots\xrightarrow{a_{N-1}}x_N$, 
$\mappingpath(\pathx{N}) = \mappingpoint(x_0)\xrightarrow{a_0} \mappingpoint (x_1)\xrightarrow{a_1}\hdots\xrightarrow{a_{N-1}} \mappingpoint(x_N)$.
Then, we can map $\strategyI^*$ to a switching strategy $\strategyx$
through
\begin{equation}\label{eq: strategy mapping}
    \strategyx^ *(\pathx{N})=\strategyI^*(\mappingpath(\pathx{N})).
\end{equation}
Further, we define the lower and upper bounds of the probability of satisfaction of $\phi$ under $\strategyI^*$ as
\begin{align}
    \satlp{q}&=\min_{\adversary \in \adversaries} \mathbb{P}[\pathI\models\phi\mid\strategyI^*, \adversary, \pathI(0)=q],\label{eq: lower optimal prob bound} \\
    \satup{q}&=\max_{\adversary\in\adversaries}\mathbb{P}[\pathI\models\phi\mid\strategyI^*,\adversary,\pathI(0)=q], \label{eq: upper optimal prob bound}
\end{align}
respectively. 
The following theorem shows that the satisfaction probability bounds also hold for Process \eqref{eq: process} under $\strategyx^*$.

\begin{theorem}\label{theorem: corectness}
Given Process \eqref{eq: process}, a compact set 
$\domain\subset \mathbb{R}^n$, and an LTLf formula $\phi$
defined over the regions of interest in $\domain$,
let $\imdp$ be the IMDP abstraction of Process \eqref{eq: process} as described in Section \ref{sec: imdp abstraction}.
Further, let $\strategyI^*$ be computed by \eqref{eq: synthesis problem} with probability bounds $\check{p}$ and $\hat{p}$ as in \eqref{eq: lower optimal prob bound} and \eqref{eq: upper optimal prob bound}, respectively. Map $\strategyI^*$ into a switching strategy $\strategyx^*$ as in \eqref{eq: strategy mapping}. Then for any initial state $x_0 \in \domain$ it holds that
\begin{equation*}
    P[\pathx{}\models\phi \ \mid \ \strategyx^*,\pathx{}(0)=x_0] \in [\satlp{\mappingpoint (x_0)},\satup{\mappingpoint(x_0)}].
\end{equation*}
\end{theorem}
The proof of this theorem follows similarly as the proof of Theorem 2 in \cite{formal2021Jackson}. 
Theorem \ref{theorem: corectness} guarantees that the probability that Process \eqref{eq: process} satisfies $\phi$ is contained in the satisfaction probability bounds $\check{p}$ and $\hat{p}$. 
The difference between $\check{p}$ and $\hat{p}$ can be viewed as the error induced by space discretization and local approximation of the NN dynamics with linear functions. This error monotonically decreases if the size of the discretization decreases. As a consequence, the synthesized strategy is optimal for an infinitely fine grid. 



\subsection{Synthesis driven refinement}\label{sec: synthesis driven refinement}
Here, we present a discretization refinement scheme that aims to efficiently reduce the error induced by the space discretization. In each refinement step, we refine a predefined fixed number of states in $\Q$, which we refer to as $n_{ref}$.  
To enable the use of Theorem \ref{theor: sorting} and Proposition \ref{prop: min log-concave func}, our refinement guarantees that all refined regions are axis-aligned hyper-rectangles in the transformed space. 
Hence, a region is refined by splitting the corresponding hyper-rectangle region $\image{q}$ over one dimension.
To decide on which states to refine, we define a score function $\theta: Q\rightarrow \mathbb{R}^+$ as
\begin{equation*}
    \theta(q) = (\satlp{q}-\satup{q}) \sum_{a\in A} \sum_{q'\in Q}\left(\uP{q'}{q}-\lP{q'}{q}\right)
\end{equation*}
where $\check{p}$ and $\hat{p}$ are the satisfaction probabilities as defined by \eqref{eq: lower optimal prob bound} and \eqref{eq: upper optimal prob bound}. We refine the $n_{ref}$ regions with the highest $ \theta(q)$. The score function serves as a measure of uncertainty caused by state $q\in \Q$ and closely resembles the uncertainty measure  proposed in \cite{coogan2018} in a verification context. 

The rationale behind our choice of which dimension to refine is based on the objective to reduce the conservatism introduced by the NN overapproximation process. 
In particular, for $q\in\Q$, we want to find the dimension that minimizes the volume of $\conv{H}$, as described in Proposition \ref{lemma: one-step crown set}, for both regions created by splitting $\image{q}$ over this dimension. 
To do so, we transform all the edges of $\image{q}$ using the bounding functions and measure the expansion of the edges, i.e., the relative difference in distance between the vertices describing an edge before and after the transformation. As $\image{q}$ is an axis-aligned hyper-rectangle, we then take the dimension to refine equal to the dimension the largest expanded edge aligns to. 
The exact procedure to find the dimension to refine can be found in 
\ifthenelse{\boolean{arxiv}}{Appendix \ref{sec: refinement procedure}}{the technical report \cite{adams2022Formal}.}



\section{Case Studies}
\begin{figure}
    \centering
    \vspace{0.5mm}
    \includegraphics[width=0.69\linewidth]{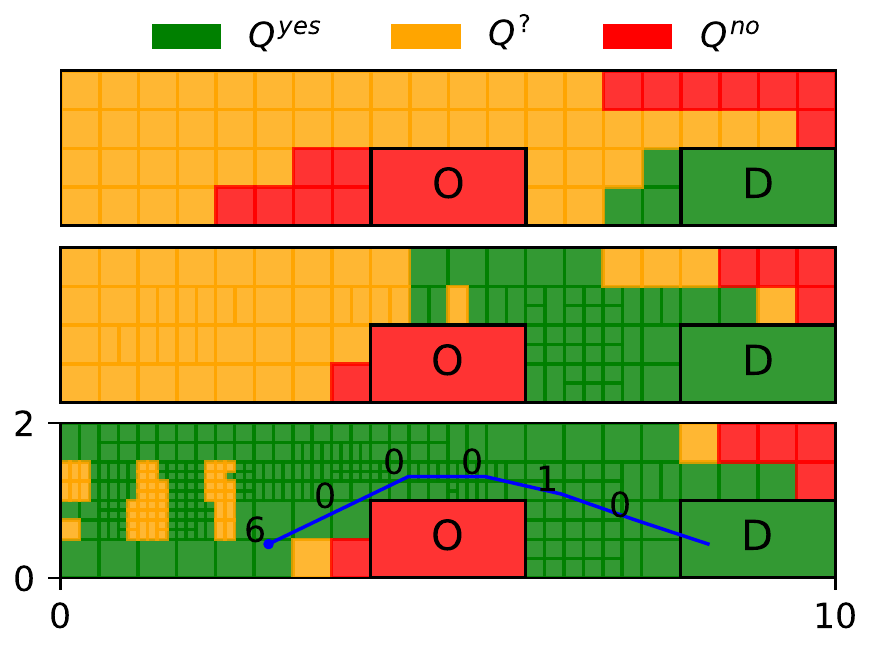} \vspace{-2mm}
    \caption{
    Region labeling and classification of each initial states $x \in X$ as $Q^{\text{yes}}$ if $\satlp{\mappingpoint(x)}\geq 0.95$, $Q^{\text{no}}$ if $\satup{\mappingpoint(x)}<0.95$ and $Q^{\text{?}}$ otherwise, for the first, an intermediate and the final abstraction of Experiment 1. In blue, a simulated path labeled with the action chosen at each step.
    }
    \label{fig: clas results car-simulations}
    \vspace{-2mm}
\end{figure}
We consider 3 different NNDMs learned on non-linear datasets taken from the literature (for details see Appendix
\ifthenelse{\boolean{arxiv}}
{\ref{sec: case studies}}{I-E in the technical report \cite{adams2022Formal}}).
All transition probability bounds of the IMDP abstractions were computed using Theorem \ref{theor: sorting}, except for the lower bounds on transitions to regions that overlap with the post image overapproximation of the region from which the transition starts. For those, we used Proposition \ref{lemma: one-step crown set}\ifthenelse{\boolean{arxiv}}{, as illustrated in Figure \ref{fig: sorting} and explained in detail in Appendix \ref{sec: algorithm}.}{.}
{Empirically, we found this approach to offer good results in balancing between precision and scalability. }
All experiments were run on an Intel Core i7-10610U CPU at 1.80GHz$\times$ 2.30Hz  with 16GB of RAM. 

\subsubsection{Efficient Control Synthesis by Iterative Refinement}\label{sec: case studies - effect refin}
We consider a 3-D car model from \cite{vehicle_dynamics_rajesh_rajamani}, with state space representing position and orientation of the car and seven discrete actions switching between different feedback controllers that interact with the car and steer it to a given orientation.
We are interested in synthesizing a strategy for a static overtaking scenario as shown in Figure \ref{fig: clas results car-simulations}. Here, the car should globally avoid an obstacle (``O'') and eventually reach a desired (``D'') region, i.e., $\phi_1 = \mathcal{G}(\lnot O) \wedge \mathcal{F}(D)$.
To do so, we start with a very coarse abstraction and iteratively refine the discretization, which overall takes approximately $36$ minutes.
From Figure \ref{fig: clas results car-simulations} we observe that the refinement procedure preserves the initial coarseness of the discretization for regions with small uncertainty on the satisfaction probability, whereas the critical regions, such as the corners around the obstacle, are further refined. Hereby, not only the lower bounds improve (the orange regions turn green), but also the upper bounds improve (red regions turn orange or green), 
and the controller strategies based on non-informative lower bounds are possibly updated (red regions turn green).
\ifthenelse{\boolean{arxiv}}{Then, for the final abstraction, which is roughly one-fifth of the number of states (approximately $10,000$) a standard uniform discretization of the domain using the finest grid of the final abstraction would contain, we obtain tight satisfaction probabilities.}{}

\subsubsection{Control Synthesis for Complex Specifications}\label{sec: complex ltl}
To show that our framework can handle complex specifications, we use four nonlinear 2-D datasets generated by the nonlinear system considered in  \cite{jackson2021strategy}, and perform control synthesis given the same labeling of the domain and complex LTLf specification as in \cite{jackson2021strategy}, i.e., $\phi_2 = \mathcal{G}(\lnot O) \wedge \mathcal{F}(D1) \wedge \mathcal{F}(D2)$. 
The iterative abstracting and control synthesis procedure takes approximately 65 seconds, and the final abstraction consists of $1,500$ states.  
Figure \ref{fig: scaling}a shows that, although we assume a noisier dataset, we are able to compute informative satisfying probability bounds that resemble the result in \cite{jackson2021strategy}.

\subsubsection{Scalability: high-dimensional and complex NN-structures}\label{sec: case studies - scalability}
Last, we test the scalability of our framework on a 5-D system with NNs of 5 hidden layers with 100 neurons per layer. Here, we consider the reach-avoid specification $\phi_1$ with the labeling of the space as shown in Figure \ref{fig: scaling}b. We again start with a coarse abstraction and iteratively improve the abstraction, which overall takes approximately 2 hours, from which  
1.5 hours for generating the IMDP abstractions. 
The final abstraction consists of approximately $15,000$ 
states. Figure \ref{fig: scaling}b shows that we are able to guarantee for a large part of the domain that the initial states almost always
(green regions) or almost never
(red regions) satisfy the specification using complex controller strategies as indicated by the simulated paths and corresponding actions. 

\begin{figure}
    \centering
    \vspace{1mm}
    \includegraphics[width=1.0\linewidth]{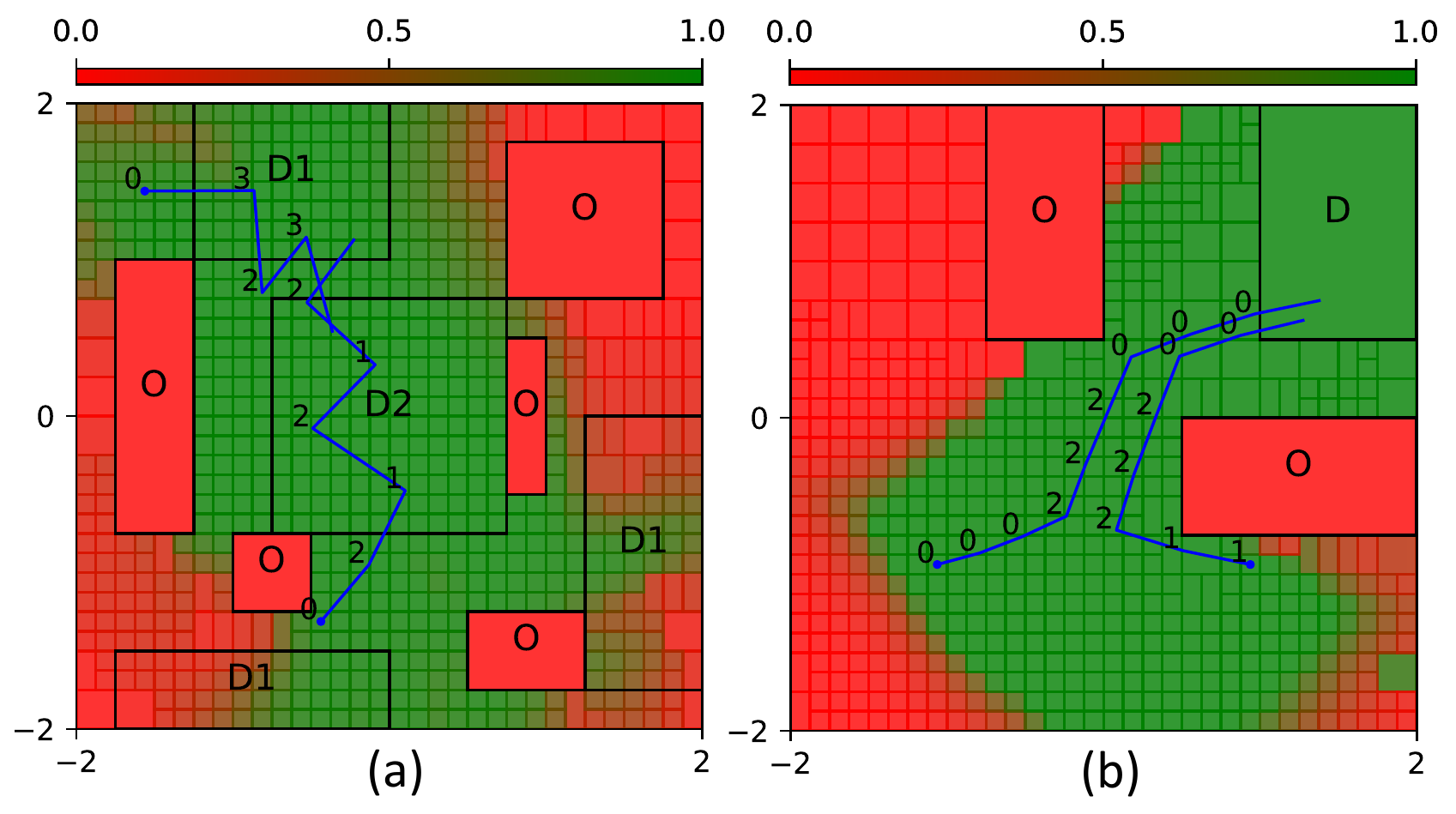}
        \caption{Region labeling and lower satisfaction probability bounds ($\check{p}$) of the initial state for Experiment 2 (a) and 3 (b). 
        }
        \label{fig: scaling}
        \vspace{-5mm}
\end{figure}

\section{Conclusion}
We introduced a formal control synthesis framework for a stochastic NNDMs with LTLf specifications.
We showed that in practice the abstraction can be constructed very
efficiently and developed an iterative refinement scheme to efficiently minimize
the number of states of this discretization-based method. 
By experiments on various datasets, we showed that our framework enables efficient control synthesis of provably correct strategies for complex NNDM of several input dimensions on nontrivial control tasks. 
\ifthenelse{\boolean{arxiv}}{In the future, we plan to extend our framework to NNDMs driven by Recurrent Neural Networks.}{}





\bibliographystyle{IEEEtran}
\bibliography{references}

\begin{thebibliography}{10}
\providecommand{\url}[1]{#1}
\csname url@samestyle\endcsname
\providecommand{\newblock}{\relax}
\providecommand{\bibinfo}[2]{#2}
\providecommand{\BIBentrySTDinterwordspacing}{\spaceskip=0pt\relax}
\providecommand{\BIBentryALTinterwordstretchfactor}{4}
\providecommand{\BIBentryALTinterwordspacing}{\spaceskip=\fontdimen2\font plus
\BIBentryALTinterwordstretchfactor\fontdimen3\font minus
  \fontdimen4\font\relax}
\providecommand{\BIBforeignlanguage}[2]{{%
\expandafter\ifx\csname l@#1\endcsname\relax
\typeout{** WARNING: IEEEtran.bst: No hyphenation pattern has been}%
\typeout{** loaded for the language `#1'. Using the pattern for}%
\typeout{** the default language instead.}%
\else
\language=\csname l@#1\endcsname
\fi
#2}}
\providecommand{\BIBdecl}{\relax}
\BIBdecl
\renewcommand{\BIBentryALTinterwordstretchfactor}{4}

\bibitem{nagabandi2018neural}
A.~Nagabandi \emph{et~al.}, ``{Neural Network Dynamics for Model-Based Deep
  Reinforcement Learning with Model-Free Fine-Tuning},'' in \emph{ICRA}, 2018.

\bibitem{Tabuada2009VerificationApproach}
P.~Tabuada, \emph{{Verification and control of hybrid systems: A symbolic
  approach}}.\hskip 1em plus 0.5em minus 0.4em\relax Springer US, 2009.

\bibitem{Lahijanian2015FormalSystems}
M.~Lahijanian \emph{et~al.}, ``{Formal Verification and Synthesis for
  Discrete-Time Stochastic Systems},'' \emph{TACON}, 2015.

\bibitem{Doyen2018VerificationSystems}
L.~Doyen \emph{et~al.}, \emph{{Verification of Hybrid Systems}}, 2018.

\bibitem{baier2008principles}
C.~Baier \emph{et~al.}, \emph{Principles of model checking}.\hskip 1em plus
  0.5em minus 0.4em\relax MIT press, 2008.

\bibitem{LTLf}
G.~De~Giacomo \emph{et~al.}, ``{Linear Temporal Logic and Linear Dynamic Logic
  on Finite Traces},'' \emph{IJCAI}, 2013.

\bibitem{raissi2019physics}
M.~Raissi \emph{et~al.}, ``{Physics-informed neural networks: A deep learning
  framework for solving forward and inverse problems involving nonlinear
  partial differential equations},'' \emph{J. Comput. Phys.}, 2019.

\bibitem{gamboa2017deep}
J.~C.~B. Gamboa, ``Deep learning for time-series analysis,'' 2017.

\bibitem{chua2018deep}
K.~Chua \emph{et~al.}, ``Deep reinforcement learning in a handful of trials
  using probabilistic dynamics models,'' in \emph{NIPS}, 2018.

\bibitem{Wei2021SafeModels}
T.~Wei \emph{et~al.}, ``{Safe Control with Neural Network Dynamic Models},''
  2021.

\bibitem{wicker2021certification}
M.~Wicker \emph{et~al.}, ``Certification of iterative predictions in bayesian
  neural networks,'' in \emph{UAI}.\hskip 1em plus 0.5em minus 0.4em\relax
  PMLR, 2021, pp. 1713--1723.

\bibitem{fazlyab2021introduction}
M.~Fazlyab \emph{et~al.}, ``An introduction to neural network analysis via
  semidefinite programming,'' in \emph{CDC}.\hskip 1em plus 0.5em minus
  0.4em\relax IEEE, 2021, pp. 6341--6350.

\bibitem{Xu2020AutomaticBeyond}
K.~Xu \emph{et~al.}, ``{Automatic Perturbation Analysis for Scalable Certified
  Robustness and Beyond},'' 2020.

\bibitem{IMDP}
R.~Givan \emph{et~al.}, ``{Bounded-parameter Markov decision processes},''
  \emph{Artificial Intelligence}, 2000.

\bibitem{Zhao2020LearningCertificates}
H.~Zhao \emph{et~al.}, ``{Learning Safe Neural Network Controllers with Barrier
  Certificates},'' \emph{FAOC}, 2020.

\bibitem{Bertsekas1996StochasticCase}
D.~P. Bertsekas \emph{et~al.}, \emph{{Stochastic optimal control: the
  discrete-time case}}, 1996.

\bibitem{ltlf_morteza}
A.~M. Wells \emph{et~al.}, ``{LTLf synthesis on probabilistic systems},'' in
  \emph{EPTCS}, 2020.

\bibitem{rockafellar2015convex}
{Ralph Tyrell Rockafellar}, \emph{{Convex analysis}}, 2015.

\bibitem{cauchi2019efficiency}
N.~Cauchi \emph{et~al.}, ``{Efficiency through uncertainty: Scalable formal
  synthesis for stochastic hybrid systems},'' \emph{HSCC}, 2019.

\bibitem{Boyd2004ConvexOptimization}
S.~Boyd \emph{et~al.}, \emph{{Convex optimization}}, 2004.

\bibitem{formal2021Jackson}
J.~Jackson \emph{et~al.}, ``{Formal Verification of Unknown Dynamical Systems
  via Gaussian Process Regression},'' 2021.

\bibitem{coogan2018}
M.~Dutreix \emph{et~al.}, ``{Efficient Verification for Stochastic Mixed
  Monotone Systems},'' in \emph{ICCPS}, 2018.

\bibitem{vehicle_dynamics_rajesh_rajamani}
{Rajesh Rajamani}, \emph{{Vehicle Dynamics and Control}}, 2011.

\bibitem{jackson2021strategy}
J.~Jackson \emph{et~al.}, ``{Strategy Synthesis for Partially-known Switched
  Stochastic Systems},'' in \emph{HSCC}, 2021.

\end{thebibliography}

\ifthenelse{\boolean{arxiv}}{
\appendices
\section{Supplementary Material}\label{sec: supplementary material}

\subsection{Efficient Computation of Transition Probabilities}
\subsubsection{Grouping Procedure}\label{sec: grouping procedure}
Recall that we discretized $\domain$ using a grid in the transformed space induced by transformation matrix $\T$. For $q\in \Q$, let vectors $\check{h}$ and $\hat{h}$ define the vertices of the hyper-rectangular overapproximation of $\rectregion{\conv{H}}$, i.e., $\rect{h} = \rectregion{\conv{H}}$, such that $\check{h}^{(l)}\leq \hat{h}^{(l)}$ for $l\in\{0,\hdots,n\}$, as in Proposition \ref{lemma: one-step crown set}. Then, for the $l$-th dimension, we denote the set of discretization intervals by
\begin{equation}
\begin{aligned}
    \Delta^{[l]} = \{[\check{z}^{(l)}, \hat{z}^{(l)}] \mid \ &\rect{z}=\image{q},\check{z}^{(l)}\leq \hat{z}^{(l)}, \\
    &q\in \Q\}
\end{aligned}
\end{equation}
Based on $\rect{h}$, we split $\Delta^{[l]}$ in three subsets, 
\begin{align}
    \Delta^{[l]}_L &= \{r \in \Delta^{[l]} \mid z<h, \forall z\in r, \forall h\in [\check{h}^{(l)},\hat{h}^{(l)}]\}\nonumber \\
    \Delta^{[l]}_O &= \{r \in \Delta^{[l]} \mid r \cup [\check{h}^{(l)},\hat{h}^{(l)}]\neq\emptyset\}\label{eq: delta_overlap S} \\
    \Delta^{[l]}_U &= \{r \in \Delta^{[l]} \mid z>h, \forall z\in r, \forall h\in [\check{h}^{(l)},\hat{h}^{(l)}]\},\nonumber
\end{align}
and, we create a new set of intervals $\Delta_{new}^{[l]}$ consisting of the union of all intervals $\Delta_L^{[l]}$, the union of all intervals in $\Delta_U^{[l]}$, and  the intervals of $\Delta_O^{[l]}$, i.e., 
\begin{equation*}
    \Delta_{new}^{[l]} = \left(\cup_{r\in \Delta^{[l]}_L}r\right) \cup \Delta_O^{[l]} \cup \left( \cup_{r\in \Delta^{[l]}_L}r \right)
\end{equation*}
By construction, it is guaranteed that each interval in $\Delta_{new}^{[l]}$ has a unique value for $z_{min}^{(l)}$ and $z_{max}^{(l)}$ as defined in Theorem \ref{theor: sorting}. 
Next, we want to define a function $G:\Q\rightarrow \Q$ that given a region $q\in \Q$ returns a grouping of the regions in $\Q$ such that all $q'\in g\in G$ have equal $z_{min}$ and $z_{max}$ as defined in Theorem \ref{theor: sorting}. 
First, we introduce a labeling function $L_{\Delta}:\mathbb{R}\times \mathbb{R}\times \mathbb{N}\rightarrow 2$, such that
\begin{equation*}
\begin{aligned}
    &L_{\Delta}(r,\delta,l) = \begin{cases}\top &\text{if}\ r\cup \delta = r \wedge
    \delta \cup [\check{h}^{(l)},\hat{h}^{(l)}]\neq \emptyset\\
    \top &\text{if} \ r \subseteq \delta \\
    \bot &\text{else}.
    \end{cases}
\end{aligned}
\end{equation*}
Then, grouping $G:\Q\rightarrow \Q$ is defined as
\begin{equation*}
\begin{aligned}
    G(q) = \{Q_{r}\subseteq\Q\mid&\forall q'\in Q_{r}, \delta\in \Delta_{new}^{[l]}, \forall l\in \{0,\hdots,n\},\\
    &L_{\Delta}(q',\delta,l) \}
\end{aligned}
\end{equation*}
where vectors $\check{h}$ and $\hat{h}$ define the vertices of the hyper-rectangular overapproximation of $\rectregion{\conv{H}}$, i.e., $\rect{h} = \rectregion{\conv{H}}$, for region $q$ as in Proposition \ref{lemma: one-step crown set}. The labeling function returns true ($\top$) for two intervals ($r$ and $\delta$) that have the same relative position w.r.t. $\conv{H}$ for dimension $l$, where the possible positions for each dimension $l$ are defined by $\Delta_{new}^{[l]}$.
Note that even for a nonuniform discretization, the labeling function defines a complete and unique mapping function, i.e., for all $q\in \Q$ it holds that $\cup_{g\in G(q)}g = \Q$ and $\forall q'\in \Q,\exists ! g\in G(q)$ s.t. $q'\in g$.

\subsubsection{Algorithm}\label{sec: algorithm}
Our approach to compute $\check{P}$ and $\hat{P}$ as described in Section \ref{sec: efficient computation} is summarized in Algorithm \ref{al:1}.

\begin{algorithm}
\DontPrintSemicolon
 \For{$q\in \Q, a\in A$}{
 Define $H$ as in Proposition \ref{lemma: one-step crown set} and construct $\rectregion{\conv{H}}$.\;
 Find grouping $G(q)$ as described in Appendix \ref{sec: grouping procedure}.\;
 \For{$Q_r \in G(q)$}{
    Pick an element $q_{ref}$ from $Q_r$.\;
    Find $z_{max}$ that upper-bounds $\tk{q_{ref}}$ for $x\in q$ as in Theorem \ref{theor: sorting}.\;
    \eIf{$q_{ref}\cup \rectregion{\conv{H}}=\emptyset$}{
    Find $\hat{v}, \check{v}$ s.t. $\rect{v}=\rectregion{q_{ref}}$\;
    Find $z_{min}$ that lower-bounds $\tk{q_{ref}}$ for $x\in q$ as in Theorem \ref{theor: sorting}.\;
    \For{$q'\in Q_r$}{
    $\check{P}(q,a,q') = g(z_{min})$ \;
    $\hat{P}(q,a,q') = g(z_{max})$ \;
    }}{
    \For{$q'\in Q_r$}{
    Find $\hat{v}, \check{v}$ s.t. $\rect{v}=\rectregion{q_{ref}}$\;
    $z_{min}=\argmin{z\in V}\g$\;
    $\check{P}(q,a,q') = g(z_{min})$\;
    $\hat{P}(q,a,q') = g(z_{max})$ \;
    }
}}}
\caption{
Computation of
$\check{P}$ and $\hat{P}$.
}\label{al:1}
\end{algorithm}

\begin{table*}[h]
\begin{tabular}{@{}llllllllll@{}} \toprule 
& $\phi$& $n$ & States &$\domain$ &$\cov{v}$ & $|A|$ & \#L & \#N/L& Proj.\\\cmidrule{1-10}
\textbf{1} &$\phi_1$ & 2& $[x_1,x_2]$&$[-2,2]\times [-2,2]$ &$0.2I$& 4 & 3 & 20& $[x_1,x_2]$ \\
\textbf{2} &$\phi_2$& 3& $[x, y, \phi]$ & $[0\times 10]\times[0\times 2]\times[-0.5\times 0.5]$ &$diag([0.1,0.1,0.01])$& 7 & 4 & 50& $[x, y]$\\
\textbf{3} &$\phi_2$&5&$[x_1,x_2,x_3,x_4,x_5]$ &$[-2,2]\times[-2,2]\times[-0.4,0.4]\times[-0.4,0.4]\times[-0.4,0.4]$ &$0.05I$&  3 & 5 & 100&$[x_1,x_2]$\\
 \bottomrule
\end{tabular}
\caption{Overview of the LTLf specification, number of dimensions of the system, state space representation, compact set, noise characteristics, number of modes, number of layers and neurons per layer of the NNs considered per experiment. Last, the states spanning the projected space as shown in the figures \ref{fig: clas results car-simulations} and \ref{fig: scaling} (Proj.), are specified.} 
\label{tab}
\end{table*}

\subsection{Refinement Procedure}\label{sec: refinement procedure}
Here, we denote by $\vertices{\rect{v}}$ the set of vertices of hyper-rectangle $\rect{v}$, i.e., 
\begin{equation}\label{eq: definition vertices}
\begin{aligned}
    \vertices{\rect{v}} = &\{([v^{(1)}, v^{(2)},\hdots, v^{(n)}]^T \mid \\
    &v^{(l)}\in \{\check{v},\hat{v}\}, l\in\{0,\hdots,n\}\}
\end{aligned}
\end{equation}

Let the vertices of hyper-rectangle $\image{q}$ be defined by vectors $\check{v}$ and $\hat{v}$ (i.e., $\rect{v}=\rectregion{q}$), and $\lbf_a,\ubf_a$ be the linear functions that bound $\nnT$ for $z\in \image{q}$ as defined by Proposition \ref{lemma: one-step crown set}, and $\check{A}_a,\hat{A}_a\in\mathbb{R}^{n\times n}$ be the matrices that capture the linear transformation of $\lbf_a$ and $\ubf_a$, respectively. 

To measure the relative deformation of the edges of $\image{q}$, 
we first denote the set of vertices of $\image{q}$ by $\vertices{q}$ as in \eqref{eq: definition vertices}. Further, we define a mapping $M_{\xi}:\mathbb{R}^{n}\rightarrow \{0,\hdots,n\}$ that returns the matching dimension(s) for two vectors $v,v'\in \mathbb{R}^n$, i.e.,
\begin{equation*}
    M_{\xi}(v,v') = l \iff v^{(l)}=v'^{(l)}
\end{equation*}
Note that, because $\image{q}$ is an axis-aligned hyper-rectangle, for all $v,v'\in \vertices{q}$ and $v\neq v'$, $M_{v}(v,v')$ is an unique and complete mapping. 
The maximum expansion of an edge that is described by vertices $v,v'\in \vertices{q}$ under transformation of $\check{A}_a$ and $\check{A}_a$ for all $a\in A$, which we refer to as $\xi_{max}(v,v')\in\mathbb{R}$, is then defined as
\begin{equation*}
    \xi_{max}(v,v') = \max_{a\in A}\left(
    \frac{\lVert\check{A}_a\check{v}-\check{A}_a\hat{v}\rVert}
    {\lVert\check{v}-\hat{v}\rVert}
    , 
    \frac{\lVert\hat{A}_a\check{v}-\hat{A}_a\hat{v}\rVert}
    {\lVert\check{v}-\hat{v}\rVert}
    ,
    \right)
\end{equation*}
Then, for $q\in \Q$ the dimension to refine, which we call $l_c\in\{0,\hdots,n\}$, is defined as the dimension corresponding to vectors $v,v'\in\vertices{q}$ that the maximize $\xi_{max}(v,v')$, i.e.,
\begin{equation*}
    l_c = M_{\xi}\left(\argmax{v, v'\in \vertices{q}} \xi_{max}(v,v') \right)
\end{equation*}

\subsection{Case Studies}\label{sec: case studies}
The details on the system and NN-characteristics considered for each system are reported in Table \ref{tab}.

}{}





\end{document}